\documentclass{amsart}

\newtheorem{theorem}{Theorem}[section]
\newtheorem{lemma}[theorem]{Lemma}

\theoremstyle{definition}
\newtheorem{definition}[theorem]{Definition}
\newtheorem{conjecture}[theorem]{Conjecture}

\theoremstyle{remark}
\newtheorem{remark}[theorem]{Remark}

\numberwithin{equation}{section}

\begin{document}
\title{On Twisted Virasoro Operators And Number Theory}
\author{An Huang}
\address{Department of Mathematics, UC Berkeley}
\email{anhuang@berkeley.edu}
\date{June 14, 2009}

\begin{abstract}
\noindent We explore some axioms of divergent series and their relations with conformal field theory. As a consequence we obtain another way of calculating $L(0,\chi)$ and $L(-1,\chi)$ for $\chi$ being a Dirichlet character. We hope this discussion is also of interest to physicists doing renormalization theory for a reason indicated in the Introduction section. We introduce a twist of the oscillator representation of the Virasoro algebra by a group of Dirichlet characters and use this to give a 'physical interpretation' of why the values of certain divergent series should be given by special L values. Furthermore, we use this to show that some fractional powers which are crucial for some infinite products to have peculiar modular transformation properties are expressed explicitly by certain linear combinations of $L(-1, \chi)$'s for appropriately chosen $\chi$'s, and can be understood physically as a kind of 'vacuum Casimir energy' in our settings. We also note a relation between field theory and our twisted operators. Lastly we give an attempt to reinterpret Tate's thesis by a sort of conformal field theory on a number field.
\end{abstract}
\maketitle
\section*{Introduction}
Some of the results of divergent series are summarized in Hardy's book  \textsl{Divergent Series} \cite{Hardy}. Hardy proposed 3 axioms on manipulating divergent series, however, his 3rd axiom is not applicable in many cases, i.e. we can't obtain answers to the series $1+1+1+...$ or $ 1+2+3+...$. In section 1, we first show how to use two other axioms (We call them axiom (1) and axiom (2) in this paper) mentioned to me by Borcherds to sum up these two series which give 'correct' answers given by zeta values. By using some analytic arguments, we prove axiom (2) alone gives values to some divergent series which agree with corresponding special Dirichlet L values. This is summarized in Theorem 1.1. Then we discuss the question of consistency of these two axioms. We also reformulate some of the results by nonstandard analysis and make a conjecture on the analytic continuation of some general Dirichlet L series to $s=0$ (conjecture 1.7). After that, we give a physical reasoning of why axiom (1) combined with axiom (2) possibly gives the 'correct' answer for $0+1+2+3+...$. To a physicist, our discussions on divergent series is possibly interesting because from a physical point of view, we are attempting to argue that 'no matter what regularization scheme one uses', if one agrees certain innocently looking axioms at work, one will always get the same answer for some divergent series including the famous $0+1+2+3+...$ that shows up a lot in physics. (for some reason, the explict discussion of this point is given at the end of section 2.) Included in our discussion also we show how amusingly axiom (2) for divergent series is related with the well known explicit Dirichlet class number formula for imaginary quadratic fields $\mathbb{Q}(\sqrt{-q})$, where $q$ is an odd prime congruent to $3$ mod $4$. These are done in section 1.

The oscillator representation of the Virasoro algebra appears in bosonic string theory as components (or modes) of the string world sheet energy momentum tensor. In this context, it is well known that there is a physical reasoning of why $1+2+3+...=-1/12$, by considering the vacuum Casimir energy of the world sheet conformal field theory. In order to get a similar physical interpretation of $\sum_{i=1}^{\infty}\chi(i)i=L(-1, \chi)$ for Dirichlet characters $\chi$, we introduce a twist of the oscillator representation by a group of Dirichlet characters. It will turn out that in this way, we get a representation of a direct sum of several copies of Virasoro algebras sharing a same central element, on the same Fock space. Consequently, we may indeed interpret $\sum_{i=1}^{\infty}\chi(i)i=L(-1, \chi)$ by considering a sort of vacuum Casimir energy just as in the well known case of $1+2+3+...=-1/12$. These will be done in section 2.

Furthermore, the 'q-trace' (or character) of the oscillator representation of the Virasoro algebra gives the essential constituent of the 1-loop vacuum partition function of the corresponding conformal field theory, which possesses certain modular transformation properties predicted by the $\operatorname{SL}(2,\mathbb{Z})$ symmetry of the defining lattice of an elliptic curve over $\mathbb{C}$. \cite{Zhu} is a fundamental paper devoted to giving a mathematical formulation and proof of such phenomenon by the theory of vertex operator algebras, which, in particular, implies that the characters of minimal model representations with negative central charges $c_{2,2k+1}$ have certain modular transformation properties. See for example \cite{Milas}. Now, in our setting of the oscillator representation twisted by a group of Dirichlet characters (or slightly more generally, a group of certain periodic functions from natural numbers to complex numbers), we show a similar story: certain classes of 'q-trace' can be expressed by certain theta functions with characteristics. Consequently, they have certain modular transformation properties by the theory of theta functions. Along the way we also recover exactly the characters of $c_{2,2k+1}$ as 'q-traces' in this different setting. The highlight is \eqref{3.28} and the discussions around it: we see those strange fractional powers which are crucial for certain infinite series to have peculiar modular transformation properties are expressed explicitly in terms of linear combinations of $\L(-1, \chi)$'s for certain $\chi$'s. They can be interpreted physically as vacuum Casimir energies in our settings, and mathematically they come from special L values. We think this advances our understanding of those fractional powers. These will be done in section 3. Our notations for theta functions will be in accordance with \cite{theta}.

By using class field theory, we canonically associate a twisted oscillator representation with a totally real finite abelian extension of $\mathbb{Q}$. From an algebraic number theoretic point of view, all the above construction is only for the rational numbers, since only Dirichlet L functions show up. A natural question is whether one can generalize some of these to more general number fields. We don't have an answer to this question, but we will try to give some hints of the difficulties involved of trying to do this in a more or less direct way. These will be done in section 4.

At the end of this paper, in section 5, we give an attempt of reinterpreting Tate's thesis by a sort of conformal field theory on a number field. We hope this paper gives some indication that there are possibly many things one can try to relate number theory and physics in this 'traditional' direction.
\begin{remark}
I should thank Antun Milas for mentioning to him that most of our constructions for twisted Virasoro operators in section 2 have been obtained independently in \cite{Bloch} or \cite{Milas2}. I received Antun's email about this immediately after the original version of this paper was posted on arxiv. Also, in \cite{Milas} Milas defined certain q-series $A_i(q)$ which are supposed to be linear combinations of Eisenstein series of weight 2 twisted by Dirichlet characters. This should be related with \eqref{3.28} as he pointed out. Also, some ideas concerning our lemma 4.3 already exist in the literature.
\end{remark}
\section{Some Axioms for Divergent Series}
Hardy's axioms for summing up divergent series (of course, they are only applicable to certain spaces of series which we will not mention) are linearity together with\\
Axiom 3 of Hardy: 
\begin{equation}
a_0+a_1+a_2+...=a_0+(a_1+a_2+...)
\end{equation}
However, with these axioms, it's very easy to see that we can't assign a finite complex number to the series $1+1+1+...$ or $1+2+3+...$.
For convenience, let's take ${b_i}$ to be the partial sum series of ${a_i}$, namely, 
\begin{equation}
b_i=a_1+a_2+...+a_i
\end{equation} 
Then $\lim b_1, b_2, b_3,...$ is another way to denote $a_0+a_1+a_2+...$.

Richard Borcherds mentioned to me the following two axioms to substitute for axiom 3 of Hardy:\\
Axiom (1)
\begin{equation*}
\lim b_1, b_2, b_3,...=\lim b_1, b_1, b_2, b_2, b_3, b_3,... 
\end{equation*}
Axiom (2)
\begin{equation*}
\lim b_1, b_2, b_3,...=\lim c_1, c_2, c_3,...               
\end{equation*}
where $c_i=\frac{b_1+b_2+...+b_i}{i}$ is the arithmetic average series of ${b_i}$. We should keep in mind that whenever a series has a finite limit in the usual sense, we take that value as the answer.\\
It should be noted that axiom (2) has been used long time ago. One may consult Cesaro summation for some details. We hope our following discussion brings some new perspective on these axioms for divergent series. But before doing anything serious, let's begin with some practice calculations with these axioms.

First, let's calculate $0+1+1+1+...$ and $0+1+2+3+...$ with these axioms:
\begin{align*}
s
&=0+1+1+1+...\\
&=\lim 0,1,2,3,...\\
&=\lim 0,0,1,1,..., \text{by axiom (1)}
\end{align*}
Now apply the linearity axiom, and use $\lim 0,1,2,3,...$ to subtract $2$ times $\lim 0,0,1,1,...$, we get
\begin{equation*}
s=-\lim 0,1,0,1,0,1,...
\end{equation*}
Apply axiom(2), we get $s=-\frac{1}{2}$.

Next, we calculate $1-2+3-4+...$ using only axiom (2) and linearity:
\begin{align*}
1-2+3-4+...
&=\lim 1, -1, 2, -2,...\\
&=\lim 1, 0, 2/3, 0, 3/5, 0,...\\
&=\frac{1}{4}, \text{by applying axiom (2) twice}
\end{align*}

The same calculation yields $s_1$=$0+1-2+3-4+...$=$\frac{1}{4}$.
Now we put axiom (1) back into play and calculate $s_2$=$0+1+2+3+...$ (This is just a heuristic calculation. Later we will discuss the problem of consistency of axioms (1) and (2) which implies a correct understanding of this calculation.):
\begin{align*}\label{**}
s_2-s_1\tag{**}
&=0+0+4+0+8+0+...\\
&=\lim 0, 0, 4, 4, 12, 12,...\\
&=\lim 0, 4, 12,...\text{by axiom (1)}\\
&=0+4+8+... \\
&=4s_2
\end{align*}
so $s_2$=$-\frac{1}{12}$.

We see from these examples that these axioms give values to certain divergent series agree with those given by special zeta values. However, with a little calculation, one finds:

axiom (2) also gives $1-4+9-16+...=\frac{1}{8}$,
But zeta values would give the answer as $0$. The explanation of this phenomenon is given in remark 1.3. Before this let's calculate by axiom (2) the value of a class of divergent series given by (special values of) Dirichlet L-series:
\begin{theorem}
Let $\chi$ be a nontrivial Dirichlet character with conductor $N$. Then axiom (2) together with linearity give values of $\sum_{i=1}^{\infty}\chi(i)$ and $\sum_{i=1}^{\infty}\chi(i)i$ agree with $L(0,\chi)$ and $L(-1,\chi)$ respectively.
\end{theorem}
\begin{remark}
In fact, we can see from the following proof by direct calculation that this is true for more general function $\chi$: $\textbf{N}\rightarrow\mathbb{C}$ having a period $N$, and satisfying $\sum_{k=1}^{N}\chi(k)=0$. Also note that we can add a finite number of zeros in front of these divergent series without affecting the result. 
\end{remark}
\begin{proof}
Let us first copy the well known formula expressing $L(1-n,\chi)$ by Bernoulli polynomials as in \cite{Hida}. Also, it is easily seen that this formula is also true for the more general $\chi$ described as in the remark above.
\begin{equation}
L(1-n,\chi)=-\sum_{a=1}^{N}\chi(a)N^{n-1}\frac{B_n(\frac{a}{N})}{n}
\end{equation}
where $B_n$ is the $n$th Bernoulli polynomial. The relevant $B_n$'s for us are:\\
$
B_0=1\\
B_1=x-\frac{1}{2}\\
B_2=x^2-x+\frac{1}{6}
$\\
So in particular, we have very explicit formulas for $L(0,\chi)$ and $L(1,\chi)$:
\begin{equation}\label{L0}
L(0,\chi)=\sum_{k=1}^{N}-\frac{k}{N}\chi(k)+\frac{1}{2}\sum_{k=1}^{N}\chi(k)
\end{equation}
\begin{equation}\label{L1}
L(-1,\chi)=\sum_{k=1}^{N}-\frac{k^2}{2N}\chi(k)+\frac{1}{2}\sum_{k=1}^{N}k\chi(k)-\frac{N}{12}\sum_{k=1}^{N}\chi(k)
\end{equation}

Next we calculate the values of these divergent series by our axioms and compare the results with the above.
\begin{align*}
\sum_{i=1}^{\infty}\chi(i)
&=\lim \chi(1),\chi(1)+\chi(2),\chi(1)+\chi(2)+\chi(3),...\\
&=\lim_{k\rightarrow\infty}\frac{k\chi(1)+(k-1)\chi(2)+...+\chi(k)}{k} \text{(We will see this limit exists in the usual sense)}
\end{align*}
Write $k=aN+b$, where $0\leq b<N$. Then it's easy to see that the above limit equals the limit of the subseries consisting only of values of $k$ which are integer multiples of $N$. Use the periodicity of $\chi$, a little calculation yields:\\
The above limit= $\lim_{a\rightarrow\infty}\frac{[\frac{(a+1)aN}{2}+a]\sum_{k=1}^{N}\chi(k)-a\sum_{k=1}^{N}k\chi(k)}{aN}$.\\
Since $\sum_{k=1}^{N}\chi(k)=0$, this equals $-\frac{\sum_{k=1}^{N}k\chi(k) }{N}$, which equals $L(0,\chi)$ by \eqref{L0}.

Now we calculate $\sum_{i=1}^{\infty}\chi(i)i$ by these axioms:
\begin{align*}
\sum_{i=1}^{\infty}\chi(i)i
&=\lim_{k\rightarrow\infty}\sum_{i=1}^{k}\chi(i)i\\
&=\lim_{l\rightarrow\infty}\frac{1}{l}\sum_{k=1}^{l}\sum_{i=1}^{k}\chi(i)i\\
&=\lim_{l\rightarrow\infty}\frac{1}{l}\sum_{i=1}^{l}\sum_{k=i}^{l}\chi(i)i\\
&=\lim_{l\rightarrow\infty}\frac{1}{l}\sum_{i=1}^{l}(l+1-i)\chi(i)i
\end{align*}
Write $l=aN+b$, where $0<b\leq N$. We will show that as $a\rightarrow\infty$, the limit $\lim_{a\rightarrow\infty}\frac{1}{N}\sum_{b=1}^{N}\frac{1}{l}\sum_{i=1}^{l}(l+1-i)\chi(i)i$ exists in the usual sense, and equals $L(-1,\chi)$. Then by our axioms we have $\sum_{i=1}^{\infty}\chi(i)i$=$L(-1,\chi)$.\\
To this end, we calculate this limit by brute force:
\begin{align*}
\frac{1}{l}\sum_{i=1}^{l}(l+1-i)\chi(i)i
&=-\frac{1}{l}\sum_{i=b+1}^{N}\sum_{k=0}^{a-1}(kN+i)^2\chi(i)-\frac{1}{l}\sum_{i=1}^{b}\sum_{k=0}^{a}(kN+i)^2\chi(i)\\&+\frac{l+1}{l}\sum_{i=b+1}^{N}\sum_{k=0}^{a-1}(kN+i)\chi(i)+\frac{l+1}{l}\sum_{i=1}^{b}\sum_{k=0}^{a}(kN+i)\chi(i)\\
&=-\frac{1}{l}\sum_{i=b+1}^{N}[\frac{1}{6}(a-1)a(2a-1)N^2+a(a-1)iN+ai^2]\chi(i)\\&-\frac{1}{l}\sum_{i=1}^{b}[\frac{1}{6}a(a+1)(2a+1)N^2+(a+1)aiN+(a+1)i^2]\chi(i)\\&+\frac{l+1}{l}\sum_{i=b+1}^{N}[\frac{1}{2}a(a-1)N+ai]\chi(i)\\&+\frac{l+1}{l}\sum_{i=1}^{b}[\frac{1}{2}(a+1)aN+(a+1)i]\chi(i)
\end{align*}
Keeping in mind to take the limit $a\rightarrow\infty$, and the assumption that $\sum_{k=1}^{N}\chi(k)=0$, we have 
\begin{align*}
\lim_{a\rightarrow\infty}\frac{1}{N}\sum_{b=1}^{N}\frac{1}{l}\sum_{i=1}^{l}(l+1-i)\chi(i)i
&=\frac{1}{N}\sum_{b=1}^{N}(b+1)\sum_{i=1}^{b}\chi(i)+\frac{1}{N}\sum_{b=1}^{N}\frac{1+b+N}{N}\sum_{i=1}^{N}i\chi(i)-\frac{1}{N}\sum_{i=1}^{N}i^2\chi(i)\\&-\frac{1}{N}\sum_{b=1}^{N}\sum_{i=1}^{b}i\chi(i)\\
&=\frac{1}{N}\sum_{i=1}^{N}\sum_{b=i}^{N}(b+1)\chi(i)+\frac{1}{N}(1+N+\frac{1+N}{2})\sum_{i=1}^{N}i\chi(i)-\frac{1}{N}\sum_{i=1}^{N}i^2\chi(i)\\&-\frac{1}{N}\sum_{i=1}^{N}\sum_{b=i}^{N}i\chi(i)\\
&=\frac{1}{N}\sum_{i=1}^{N}\frac{(N+i+2)(N-i+1)}{2}\chi(i)+\frac{1}{N}(1+N+\frac{1+N}{2})\sum_{i=1}^{N}i\chi(i)\\&-\frac{1}{N}\sum_{i=1}^{N}i^2\chi(i)-\frac{1}{N}\sum_{i=1}^{N}(N+1-i)i\chi(i)\\
&=-\frac{1}{2N}\sum_{i=1}^{N}i^2\chi(i)+\frac{1}{2}\sum_{i=1}^{N}i\chi(i)
\end{align*}
Again, the right hand side of the last equality equals $L(-1,\chi)$ by \eqref{L1}.
\end{proof}

The above proof doesn't give any implication why this theorem is true, and it relies on the well known formulas for special L values. Next, we will give a completely different sketch proof of theorem 1.1 which tells us the reason why those axioms work, and also indicates why in some sense, $L(-1,\chi)$ is the best we can do with these axioms. 
\begin{proof}[Sketch of Second Proof]
The Dirichlet series for $\chi$ is convergent when $\Re{s}>0$. The crucial trick is to 'enlarge' it's domain of convergence a bit further to the left by taking arithmetic average.\\
When $\Re(s)>0$, we have:
\begin{align*}
L(s,\chi)
&=\chi(1)1^{-s}+\chi(2)2^{-s}+\chi(3)3^{-s}+...\\
&=\lim \chi(1)1^{-s},\chi(1)1^{-s}+\chi(2)2^{-s},\chi(1)1^{-s}+\chi(2)2^{-s}+\chi(3)3^{-s},...\\
&=\lim_{l\rightarrow\infty}\frac{1}{l}\sum_{k=1}^{l}\frac{(l+1-k)\chi(k)}{k^s}
\end{align*}
We will prove that the limit on the right hand side exists for $\Re{s}>-1$. 

It's more or less obvious that we can take the subseries for $l=aN$ without affecting this limit. Then the limit becomes:
\begin{equation*}
\lim_{a\rightarrow\infty}\frac{1}{aN}\sum_{i=1}^{N}\sum_{k=0}^{a-1}\frac{(aN+1-kN-i)\chi(i)}{(kN+i)^s}
\end{equation*}
Using the Euler-Maclaurin formula, we may expand $\sum_{i=1}^{N}\sum_{k=0}^{a-1}\frac{(aN+1-kN-i)\chi(i)}{(kN+i)^s}$ as an asymptotic series as powers of $a$, and observe that it's enough to show that terms of power higher than $a^{-s}$ in this expansion vanishes. Namely, it's enough to check that the coefficients of $a^{2-s}$ and $a^{1-s}$ both vanish. The former is quite easy provided that we keep in mind the equality $\sum_{k=1}^{N}\chi(k)=0$. For the later, we write
\begin{equation*}
\sum_{i=1}^{N}\sum_{k=0}^{a-1}\frac{(aN+1-kN-i)\chi(i)}{(kN+i)^s}
=\sum_{i=1}^{N}\sum_{k=0}^{a-1}[\frac{(aN+1)\chi(i)}{(kN+i)^s}-\frac{\chi(i)}{(kN+i)^{s-1}}]
\end{equation*}
Expanding these two terms using the Euler-Maclaurin formula, the relevant terms which have possibly nontrivial contributions to the coefficient of $a^{1-s}$ are:
\begin{equation*}
\sum_{i=1}^{N}[\frac{aN+1}{N^s}\frac{1}{1-s}(a+\frac{i}{N}-1)^{1-s}\chi(i)-\frac{1}{N^{s-1}}\frac{1}{2-s}(a+\frac{i}{N}-1)^{2-s}\chi(i)]
\end{equation*}
which gives a total contribution to the coefficient of $a^{1-s}$ as:
\begin{equation*}
\sum_{i=1}^{N}[\frac{N}{N^s}\frac{1}{1-s}(1-s)(\frac{i}{N}-1)\chi(i)-\frac{1}{N^{s-1}}\frac{1}{2-s}(2-s)(\frac{i}{N}-1)\chi(i)]=0.
\end{equation*}

Note that the above argument already proved the first part, i.e. $L(0,\chi)$ part of theorem 1.1, since $0$ is in the domain $\Re(s)>-1$. To prove the part for $L(-1,\chi)$, one needs to show in addition that the arithmetic average of the series $\left\{\frac{1}{l}\sum_{k=1}^{l}\frac{(l+1-k)\chi(k)}{k^s}\right\}$, regarded as a real function of $s$, is right continuous at $s=-1$. (Since we have already shown it's well defined at $s=-1$ in the first proof.) We believe this can be done routinely, and by using some elementary properties of the Riemann zeta function. However, we will not honestly include this calculation here, since it's not illuminating for the rest of our discussion.
\end{proof}
\begin{remark}
If one wishes, one may calculate the coefficient of $a^{-s}$, and see that it's nonzero in general. That's why $L(-1,\chi)$ is the best we can do in some sense.
\end{remark}
As a digression, let's see how amusingly these axioms for divergent series are related with the explicit Dirichlet class number formula for quadratic imaginary fields $\mathbb{Q}(\sqrt{-q})$, where $q$ is an odd prime congruent to $3$ mod $4$.

Let $K$ denote the quadratic imaginary field $\mathbb{Q}(\sqrt{-q})$, $\zeta_K(s)$ the Dedekind zeta function for $K$, and $\zeta(s)$ the Riemann zeta function. Let $\chi$ now be the unique Dirichlet character given by the quadratic residue symbol of $q$. Then the quadratic reciprocity law implies the relation $\zeta(s)L(s,\chi)$=$\zeta_K(s)$. From this relation and the analytic class number formula, we have the well known (for the case $q$ being an odd prime congruent to $3$ mod $4$):
\begin{equation*}
L(1,\chi)=\frac{2\pi h(K)}{w\sqrt{q}}
\end{equation*}
where $h(K)$ denotes the class number of $K$, and $w$ the number of roots of unity in $K$, which is $2$ if $q>3$, and $6$ if $q=3$. For simplicity of illustration, we consider the case $q>3$ in the following. (and the case of $q=3$ is surely no more difficult, and the reader can do it himself or herself.)\\
Furthermore, the functional equation for Dirichlet L functions gives the following:
\begin{equation*}
L(0,\chi)=-i\tau(\chi)\pi^{-1}L(1,\chi)
\end{equation*}
where $\tau(\chi)$ is the Gauss sum attached to $\chi$, and in our case equals $i\sqrt{q}$.\\
Combining these two formulas and \eqref{L0}, we get:
\begin{equation*}
h(K)=L(0,\chi)=-\frac{1}{q}\sum_{k=1}^{q-1}k(\frac{k}{q})
\end{equation*} 
the explicit Dirichlet class number formula.

In other words, $h(K)$ equals $L(0,\chi)$, which can be calculated by either \eqref{L0}, or our axioms of divergent series. The first part of our second proof of theorem 1.1 can also be regarded as an independent proof of this explicit class number formula.

If we forget everything about Dirichlet L functions, it seems rather surprising that axiom (2) for divergent series is related with class number. For example, pick $q=7$, then axiom (2) for divergent series gives the following 'formula':
\begin{equation*}
\text{class number of} \hspace{.3cm} \mathbb{Q}(\sqrt{-7})=1-1-1+1-1-1+0+1-1-1+1-1-1+0+...=1
\end{equation*}
If $\chi$ is a trivial character of conductor $N$, we may use an extended version of axiom (1) combined with axiom (2) to calculate the values of $\sum_{i=1}^{\infty}\chi(i)$ and $\sum_{i=1}^{\infty}\chi(i)i$.\\
Axiom $(1)^{'}$: For any natural number $k$, 
\begin{equation*}
\lim b_1, b_2, b_3,...=\lim b_1, b_1,..., b_1, b_2, b_2,..., b_2, b_3, b_3,..., b_3,...
\end{equation*}
where each $b_i$ appears in the second series $k$ times.

Before we do any serious calculations with axioms (2) and $(1)^{'}$ (such as applying them together to calculate $0+1+2+3+...$, as we have shown heuristically), we first discuss the question of the consistency of these two axioms. They are not always consistent with each other: it's not hard to see that if one applies these two axioms in some different orders to the divergent series $0+1+1+1+...$, one can get different values:\\
Apply axiom (2) and linearity, one gets
\begin{align*}
0+1+1+1+...
&=\lim 0,1,2,...\\
&=\lim 0,\frac{1}{2},1,...
\end{align*}
So $0+1+1+1+...$ equals half itself, and so it has to be equal to $0$. But as we calculated before by applying axioms (1) and (2) and linearity with another order, it equals $-\frac{1}{2}$.

However, if we restrict the applicability of these axioms in the following way, they are indeed consistent with each other (we must point out that the following restriction is not natural. On the other hand, we will see indeed that there is a natural way based on these two axioms, to obtain results like $0+1+2+3+...=-\frac{1}{12}$. The explanation of this will be given at the end of section 2.):

(I) axiom (2) applies to $0+\sum_{i=1}^{\infty}\chi(i)$ or $0+\sum_{i=1}^{\infty}\chi(i)i$ if and only if $\chi$ satisfies the conditions in Remark 1.2. i.e. $\chi$: $\textbf{N}\rightarrow\mathbb{C}$ having a period $N$, and satisfying $\sum_{k=1}^{N}\chi(k)=0$. Furthermore, once axiom (2) is applied to any particular series, one has to apply only axiom (2) to whatever resulting series in consecutive steps until one gets the answer to this series.

(II) axiom $(1)^{'}$ applies to the series $0+1\cdot a_1+2\cdot a_2+3\cdot a_3+...$ if and only if the Dirichlet L series $\frac{a_1}{1^s}+\frac{a_2}{2^s}+\frac{a_3}{3^s}+...$ is convergent when $\Re(s)$ is large enough, and it has a meromorphic analytic continuation to the whole complex plane $L(s,\chi)$, and $L(s,\chi)$ is analytic at $s=-1$.
\begin{lemma}
Axiom $(1)^{'}$, axiom (2), and the axioms of linearity are consistent with each other provided that we put the above restrictions (I) and (II) on the applicability of them.
\end{lemma}
\begin{proof}
We denote by $V$ the complex vector space of divergent series $0+1\cdot a_1+2\cdot a_2+3\cdot a_3+...$ satisfying conditions for axiom $(1)^{'}$. Obviously the domain on which axiom (2) applies is a subspace of $V$. For any element in $V$ of the above form, We define the value of this series to be the corresponding special L value at $s=-1$. So this gives us a single valued function on $V$. Let's denote this function by $f$. (Note that this definition is the same as saying that we define the value of the divergent series $0+a_1+a_2+a_3+...$ to be the corresponding special L value at $s=0$. ) Obviously $f$ satisfies the linearity axiom. Furthermore it is straightforward to show that this definition satisfies axiom $(1)^{'}$. Moreover, the first proof of theorem 1.1 and remark 1.2 show that this definition also satisfies axiom (2) for the case when we can apply axiom (2).

So for any divergent series $0+1\cdot a_1+2\cdot a_2+3\cdot a_3+...$, if one gets a value to it by applying a finite sequence of these axioms and linearity, then this divergent series is an element in $V$, and one gets a finite subset $X$ of $V$, and a finite system of linear equations satisfied by a single valued function $f_1$ on $X$, and a unique solution. By what we have said above, we may switch $f_1$ with $f$, and get exactly the same system of linear equations. So $f$ and $f_1$ has to agree on $X$, and so on this chosen divergent series. So we can only possibly get one value to our chosen divergent series by using our axioms, and this value agrees with the one given by $f$. This implies that axioms (2) and $(1)^{'}$ and linearity are consistent with each other.
\end{proof}
On the other hand, there is a way to use axiom $(1)^{'}$ combined with axiom (2) and linearity to give values to  $\sum_{i=1}^{\infty}\chi(i)$ and $\sum_{i=1}^{\infty}\chi(i)i$ for $\chi$ a function satisfying conditions in remark 1.2, or a trivial Dirichlet character, or some other similar functions we don't discuss here (agree with $L(0,\chi)$ and $L(1,\chi)$, of course): the former case is covered by theorem 1.1, and for the later, one first check for the case when $N$ is prime, and then the proof for the general case requires just a little more effort. So, our axioms can be regarded as an algebraic way to get some results of analytic continuations for Dirichlet L function. (However, these axioms are not really purely algebraic, since the domain of applicability of axiom $(1)^{'}$ is not described algebraically. ) It is interesting to think of the question of finding algebraic axioms for divergent series corresponding to some other special L values, such as $L(-2,\chi)$, $L(-3,\chi)$, etc. Note that axiom $(1)^{'}$ still applies to these cases, but axiom (2) doesn't.

Inspired by Borcherds, we try to formulate some of the above results in terms of nonstandard analysis, and make a technical conjecture about analytic continuation of Dirichlet L series afterwards. One can consult any textbook on nonstandard analysis for basic terminologies.

Again, let $\left\{b_n\right\}$ be a sequence of complex numbers, and $\left\{c_n\right\}$ be the arithmetic average sequence of $\left\{b_n\right\}$. We choose a hyperreal number system and make the following definitions for the limit of $\left\{b_n\right\}$:
\begin{definition}
If $\left\{b_n\right\}$ is bounded, then we say $\lim b_n$ exists if and only if $[c_N]$ is the same real number for all infinite integers $N$ which are divisible by any finite integer, where $[c_N]$ denotes the standard part of the finite hyperreal number $c_N$. If $\lim b_n$ exists, we define its value to be $[c_N]$.
\end{definition}
\noindent Obviously if the sequence $\left\{b_n\right\}$ has a finite limit in the usual sense, our definition for $\lim b_n$ agrees with the usual one. Moreover we have the following:
\begin{lemma}
the above definition satisfies axiom $(1)^{'}$ for bounded sequences.
\end{lemma}
\begin{proof}
Choose any natural number $k$, obviously the map $N\rightarrow N/k$ from the set of infinite integers divisible by any finite integer to itself is one-to-one and onto. $c_N$ for the sequence $b_1, b_1,..., b_1, b_2, b_2,..., b_2, b_3, b_3,..., b_3,...$ (where each $b_i$ appears $k$ times) is the same as $c_{N/k}$ for the sequence $\left\{b_n\right\}$, for all infinite integers $N$ which are divisible by any finite integer. So the set of $c_N$ for the former inflated sequence is the same as the set of $c_N$ for the sequence $\left\{b_n\right\}$ (where $N$ runs through all infinite integers divisible by any finite integer), and so one has a limit if and only if the other has one, and the values of limits are the same if they exist.
\end{proof}
\noindent Furthermore, theorem 1.1 tells us that for functions $\chi$ satisfying conditions in remark 1.2, definition 1.5 gives a unique value to the divergent series $\sum_{i=1}^{\infty}\chi(i)$ which agrees with $L(-1,\chi)$. This inspires us to make the following conjecture on the analytic continuation of Dirichlet L series: let $a_1,a_2,a_3,...$ be a series of complex numbers such that its partial sum series $\left\{b_n\right\}$ is bounded.
\begin{conjecture}
If $\lim b_n$ exists as defined in definition 1.5, and its existence and value are independent of our choice of hyperreal number system, then the Dirichlet L series $\sum_{n=1}^{\infty}\frac{a_n}{n^s}$ can be analytically continuated to $s=0$ with the single value given by $\lim b_n$.
\end{conjecture}
\begin{remark}
Lemma 1.6 is a support to the above conjecture, since as we discussed in the proof of lemma 1.4: the function $f$ in that proof, which roughly speaking is to define the value of certain divergent series by analytic continuation of Dirichlet L series, also satisfies axiom $(1)^{'}$. However, we don't have strong support for this conjecture, and we think it's quite possible that some variation of the conjecture is correct, if this conjecture is to make sense after all.
\end{remark}     
From theorem 1.1, we have understood why axiom (2) alone give values agree with special L values associated to some divergent series. Our next aim is to consider the divergent series $0+1+2+3+...$ in terms of the conformal field theory of free scalar fields, and from there we give a physical interpretation of axiom $(1)^{'}$ and why this axiom possibly leads to the answer $-1/12$. (This also gives us a reason to replace axiom (1) by the more powerful axiom $(1)^{'}$.)

In bosonic string theory, we use scalar fields $X^\mu$ to describe the embedding of the string world sheet into background spacetime. For introductory reference, see \cite{Polchinski}. These fields are described by a conformal field theory on the world sheet. Consider the 'holomorphic part' of the theory, we have the world sheet energy momentum tensor whose components give rise to the oscillator representation of the Virasoro algebra with central charge $c=1$:
\begin{equation}
L_n:=\frac{1}{2}\sum_{j\in\mathbb{Z}}:a_{-j}a_{j+n}:
\end{equation}
where $a_j$'s are operators representing the oscillator algebra. i.e. 
\begin{equation*}
[a_m,a_n]=m\delta_{m,-n}
\end{equation*}
Where 
\begin{equation*}
:a_ia_j:=
\begin{cases}
a_ia_j &\text{if $i\le j$,}\\
a_ja_i &\text{otherwise}
\end{cases}
\end{equation*}
is the creation-annihilation normal ordering.

Note that in particular,
\begin{align*}
L_0&=\frac{1}{2}\sum_{j\in\mathbb{Z}}:a_{-j}a_{j}:\\
&=\frac{1}{2}a_{0}^{2}+\frac{1}{2}\sum_{j=1}^{\infty}a_{-j}a_{j}+\frac{1}{2}\sum_{j=1}^{\infty}:a_{j}a_{-j}:
\end{align*}
If we calculate the classical energy momentum tensor by variation of the world sheet action against the world sheet metric, we will get the zero mode of the classical energy momentum tensor before quantization as:
\begin{equation}
L_{0}^{c}=\frac{1}{2}\sum_{j=-\infty}^{\infty}a_{-j}a_{j}
\end{equation}
Formally, if we use the commutation relations of the oscillator algebra to pass from the classical $L_0^c$ to quantum $L_0$, we get
\begin{align*}
L_{0}^{c}
&=\frac{1}{2}a_{0}^2+\sum_{j=1}^{\infty}a_{-j}a_{j}+\frac{1}{2}(0+1+2+3+...)\\
&=L_0+\frac{1}{2}(0+1+2+3+...)
\end{align*}
Classically, we have $L_{m}^{c}=0$ as the equation of motion of the world sheet metric. In particular, $L_{0}^{c}=0$. So formally, this leads us to the requirement:
\begin{equation}
(L_0+\frac{1}{2}(0+1+2+3+...))v=0
\end{equation}
for all physical states $v$.

This gives rise to a contribution of this oscillator representation of the Virasoro algebra to the vacuum energy formally as $\frac{1}{2}(0+1+2+3+...)$. On the other hand, self-consistency of the conformal field theory gives other ways to calculate the value of this vacuum energy, and giving the value to be $-\frac{1}{24}$. For reference, see for example\cite{Polchinski}, Page 54, and page 73. Mathematically, the outcome of this physics is that this is the value by which a shift of $L_0$ eliminates the linear term in the Virasoro algebra commutation relations among the $L_{m}$'s. So this piece of physics requires us to assign value $-\frac{1}{12}$ to the divergent series $0+1+2+3+...$. In the next section, we will see this generalizes to giving physical interpretations to $\sum_{i=0}^{\infty}\chi(i)i$=$L(-1,\chi)$ for $\chi$ as in theorem 1.1.

Now, we are in position to give a physical interpretation of our axiom $(1)^{'}$ for the divergent series $0+1+2+3+...$.

Being a free scalar field, the $X^{\mu}$ conformal field theory happens to have a scaling symmetry. This can be seen from the spacetime propagator of the field $X^{\mu}$ or, for example, the equation 2.7.11 on page 60 of \cite{Polchinski}. For convenience, we'll copy this equation here:
\begin{equation}
X^{\mu}(z,\overline{z})X^{\nu}(z',\overline{z}')=:X^{\mu}(z,\overline{z})X^{\nu}(z',\overline{z}'):-\frac{\alpha'}{2}\eta^{\mu\nu}\ln{|z-z'|^2}
\end{equation}
This equation appears slightly different than the one in the book, but they are actually the same. From which we can see, if we make the change $(z-z')\rightarrow(z-z')^n$, and divide by $2n$, the expression $-\frac{\alpha'}{2}\eta^{\mu\nu}\ln{|z-z'|^2}$ remains unchanged. This is the scaling symmetry we will talk about.

Now, for $L_0=\frac{1}{2}\sum_{j\in\mathbb{Z}}:a_{-j}a_{j}:$, we have defined the value of the divergent series $\frac{1}{2}(0+1+2+3+...)$ to be the vacuum Casimir energy. Under this scaling symmetry, we have the corresponding self-embedding of the oscillator algebra $\tau_{l}$: $a_k\rightarrow a_{lk}$, $1\rightarrow l$, for any positive integer $l$. This self-embedding extends to an endomorphism of the universal eveloping algebra of the oscillator algebra. So consequently the operators $\frac{1}{l}\tau_l(L_{m})$ satisfy the same Virasoro commutation relations as before, which is also predicted by the scaling symmetry of the physics. As the physics should remain unchanged, we should have the same vacuum Casimir energy as before. But this time, it is given by the divergent series:
\begin{equation*}
\frac{1}{2}(0+0+...+0+1+0+...+0+2+0+...+0+3+0+...+0+...)
\end{equation*} 
where we have $l-1$ many copies of zero between every consecutive pairs of numbers. So the value of this divergent series should be the same as $\frac{1}{2}(0+1+2+3+...)$, and this is exactly what axiom $(1)^{'}$ says about the series $\frac{1}{2}(0+1+2+3+...)$.

This completes the physical interpretation of axiom $(1)^{'}$ for the series $0+1+2+3+...$, and why it possibly leads to the answer given by special zeta value when combined with axiom (2). (because this should agree with the value given by the consistency of the physics.) Our next step here is to contemplate on some further implications of this. From the discussion before, we know that axiom $(1)^{'}$ combined with axiom (2) give values to $\sum_{i=1}^{\infty}\chi(i)i$ coincide with $L(-1,\chi)$, for any Dirichlet character. This lets us think of the possibility of trying to introduce Dirichlet characters in a modified version of the free scalar conformal field theory, so we can possibly get a physical interpretation of axiom $(1)^{'}$ also for the series $\sum_{i=0}^{\infty}\chi(i)i$, for any Dirichlet character $\chi$. So next, we will do this to achieve our goal here, and to go deeper later.

\section{Twisted Virasoro Operators}
Let $N$ be an odd positive integer, and $G$ be a finite abelian group of functions $\mathbb{Z}/N\mathbb{Z}\rightarrow\mathbb{C}$ satisfying:\\
For any $\chi\in G$,
\begin{equation}\label{2.1}
\chi(j)=\chi(-j)
\end{equation}
for any $j$, and for any $\chi \in G$ which is not the unit,
\begin{equation}
\sum_{k=1}^{N}\chi(k)=0
\end{equation}
Then we define operators 
\begin{equation}\label{2.3}
L_{n}^{\chi}:=\frac{1}{2N}\sum_{j\in\mathbb{Z}}\chi(j):a_{-j}a_{j+nN}:
\end{equation}
for $\chi\in G$, and $n\in \mathbb{Z}$.

First we will show that these operators are closed under the Lie bracket. For this purpose, we need two little lemmas:
\begin{lemma}
For $m=Nk$, $$\sum_{j=-1}^{-m}\chi(j)j=m(L(0,\chi)-\frac{1}{2}\sum_{k=1}^{N}\chi(k))-\frac{\sum_{k=1}^{N}\chi(k)}{2}m(\frac{m}{N}-1)$$
\end{lemma}
\begin{proof}
\begin{align*}
\sum_{j=-1}^{-m}\chi(j)j
&=-\sum_{j=1}^{m}\chi(j)j\\
&=(-1)\sum_{s=0}^{k-1}\sum_{j=sN+1}^{(s+1)N}\chi(j)j\\
&=kN(L(0,\chi)-\frac{1}{2}\sum_{k=1}^{N}\chi(k))-aN(\frac{1}{2})k(k-1)\\
&=m(L(0,\chi)-\frac{1}{2}\sum_{k=1}^{N}\chi(k))-\frac{\sum_{k=1}^{N}\chi(k)}{2}m(\frac{m}{N}-1)
\end{align*}
\end{proof}
\begin{lemma}
For $m=Nk$,
\begin{align*} \sum_{j=-1}^{-m}\chi(j)j^2&=\frac{1}{6}(m-N)m(\frac{2m}{N}-1)\sum_{k=1}^{N}\chi(k)-(m-N)m(L(0,\chi)-\frac{1}{2}\sum_{k=1}^{N}\chi(k))\\&-2m(L(-1,\chi)+\frac{N}{2}L(0,\chi)-\frac{N}{6}\sum_{k=1}^{N}\chi(k))
\end{align*}
\end{lemma}
\begin{proof}
\begin{align*}
\sum_{j=-1}^{-m}\chi(j)j^2
&=\sum_{r=0}^{k-1}\sum_{j=1}^{N}\chi(j)(rN+j)^2\\
&=\sum_{r=0}^{k-1}\sum_{j=1}^{N}[r^2N^2\chi(j)+2rNj\chi(j)+j^2\chi(j)]\\
&=\frac{1}{6}(k-1)k(2k-1)N^2a+(k-1)kN\sum_{j=1}^{N}\chi(j)j+k\sum_{j=1}^{N}\chi(j)j^2\\
&=\frac{1}{6}(k-1)k(2k-1)N^2a+(k-1)kN[-N(L(0,\chi)-\frac{1}{2}\sum_{k=1}^{N}\chi(k))]+\\&k(-2N)[L(-1,\chi)+\frac{N}{2}L(0,\chi)-\frac{N}{6}\sum_{k=1}^{N}\chi(k)]\\
&=\frac{1}{6}(m-N)m(\frac{2m}{N}-1)\sum_{k=1}^{N}\chi(k)-(m-N)m(L(0,\chi)-\frac{1}{2}\sum_{k=1}^{N}\chi(k))\\&-2m(L(-1,\chi)+\frac{N}{2}L(0,\chi)-\frac{N}{6}\sum_{k=1}^{N}\chi(k))
\end{align*}
\end{proof}
Having these two computational lemmas, let's calculate $\sum_{j=-1}^{-m}\chi(j)j(m+j)$:
\begin{align*}
\sum_{j=-1}^{-m}\chi(j)j(m+j)
&=m^2(L(0,\chi)-\frac{1}{2}\sum_{k=1}^{N}\chi(k))-\frac{\sum_{k=1}^{N}\chi(k)}{2}m^2(\frac{m}{N}-1)+\frac{1}{6}(m-N)m(\frac{2m}{N}-1)\sum_{k=1}^{N}\chi(k)\\&-(m-N)m(L(0,\chi)-\frac{1}{2}\sum_{k=1}^{N}\chi(k))-2m(L(-1,\chi)+\frac{N}{2}L(0,\chi)-\frac{N}{6}\sum_{k=1}^{N}\chi(k))\\
&=-2mL(-1,\chi)+(-\frac{m^3}{6N})\sum_{k=1}^{N}\chi(k)
\end{align*}
So we have
\begin{equation}\label{2.4}
-\frac{1}{2}\sum_{j=-1}^{-m}\chi(j)j(m+j)=mL(-1,\chi)+\frac{m^3}{12N}\sum_{k=1}^{N}\chi(k)
\end{equation}
Next, we calculate the commutators of the $L_n^{\chi}$'s routinely, and we proceed as in \cite{Kac}.
\begin{lemma}
\begin{equation}
[a_k,L_n^{\chi}]=\frac{1}{N}\chi(k)ka_{k+nN}
\end{equation}
\end{lemma}
\begin{proof}
Define the function $\psi$ on $\mathbb{R}$ by:
\begin{equation*}
\psi(x)=
\begin{cases}
1 &\text{if $|x|\leq 1$,}\\ 0 &\text{if $|x|>1$}
\end{cases}
\end{equation*}
Put
\begin{equation}
L_n^{\chi}(\epsilon)=\frac{1}{2N}\sum_{j\in\mathbb{Z}}\chi(j):a_{-j}a_{j+nN}:\psi(\epsilon j)
\end{equation}
Note that $L_n^{\chi}(\epsilon)$ contains only a finite number of terms if $\epsilon\neq 0$ and that $L_n^{\chi}(\epsilon)\rightarrow L_n^{\chi}$ as $\epsilon\rightarrow 0$. More precisely, the latter statement means that, given any $v$ in the Fock space, $L_n^{\chi}(\epsilon)(v)=L_n^{\chi}(v)$ for $\epsilon$ sufficiently small.

$L_n^{\chi}(\epsilon)$ differs from the same expression without normal ordering by a finite sum of scalars. This drops out of the commutator $[a_k,L_n^{\chi}(\epsilon)]$ and so
\begin{align*}
[a_k,L_n^{\chi}(\epsilon)]
&=\frac{1}{2N}\sum_{j\in\mathbb{Z}}[a_k,\chi(j)a_{-j}a_{j+nN}]\psi(\epsilon j)\\
&=\frac{1}{2N}\sum_{j\in\mathbb{Z}}[a_k,\chi(j)a_{-j}]a_{j+nN}\psi(\epsilon j)+\frac{1}{2N}\sum_{j\in\mathbb{Z}}a_{-j}[a_k,\chi(j)a_{j+nN}]\psi(\epsilon j)\\ 
&=\frac{1}{2N}\chi(k)ka_{k+nN}\psi(\epsilon k)+\frac{1}{2N}\chi(-nN-k)ka_{k+nN}\psi(\epsilon (k+nN))
\end{align*}
Since $\chi(-nN-k)=\chi(k)$, the $\epsilon\rightarrow 0$ limit gives the result of the lemma.
\end{proof}
Next we calculate the commutator $[L_m^{\chi_1},L_n^{\chi_2}]$, and the result is the following theorem:
\begin{theorem}
\begin{equation}
[L_m^{\chi_1},L_n^{\chi_2}]=(m-n)L_{m+n}^{\chi_1\chi_2}+\delta_{m,-n}[\frac{m}{N}L(-1,\chi_1\chi_2)+\frac{m^3}{12}\sum_{k=1}^{N}(\chi_1\chi_2)(k)]
\end{equation}
\end{theorem}
\begin{proof}
For notational simplicity, we denote $\chi_1\chi_2$ by $\omega$. We have
\begin{align*}
[L_m^{\chi_1}(\epsilon),L_n^{\chi_2}]
&=\frac{1}{2N}\sum_{j\in\mathbb{Z}}[\chi_1(j)a_{-j}a_{j+mN},L_n^{\chi_2}]\psi(\epsilon j)\\
&=\frac{1}{2N}\sum_{j\in\mathbb{Z}}\chi_1(j)[\frac{1}{N}\chi_2(-j)(-j)a_{-j+nN}a_{j+mN}+\frac{1}{N}\chi_2(j+mN)(j+mN)a_{-j}a_{j+mN+nN}]\psi(\epsilon j)\\
&=\frac{1}{2N^2}\sum_{j\in\mathbb{Z}}\omega(j)[(-j)a_{-j+nN}a_{j+mN}\psi(\epsilon j)+(j+mN)a_{-j}a_{j+(m+n)N}\psi(\epsilon j)]
\end{align*}
We split the first sum into terms satisfying $j\geq \frac{(n-m)N}{2}$ which are in normal order and reverse the order of terms for which $j<\frac{(n-m)N}{2}$ using the commutation relations. In the same way we split the second sum into terms satisfying $j\geq -\frac{(n+m)N}{2}$ and $j< -\frac{(n+m)N}{2}$ . Then
\begin{align*}
[L_m^{\chi_1}(\epsilon),L_n^{\chi_2}]&
=\frac{1}{2N^2}\sum_{j\in\mathbb{Z}}\omega(j)[(-j):a_{-j+nN}a_{j+mN}:\psi(\epsilon j)+(j+mN):a_{-j}a_{j+(m+n)N}:\psi(\epsilon j)]\\&-\frac{1}{2N^2}\sum_{j=-1}^{-mN}(j+mN)j\omega(j)
\end{align*}
Making the transformation $j\rightarrow j+nN$ in the first sum and taking the limit $\epsilon \rightarrow 0$, and using \eqref{2.4}, we get the desired result.
\end{proof}
So in particular, we get a Lie algebra from 'twisting' the operators $L_m$ by a finite abelian group of functions $G$. Now we are in position to give the promised physical interpretation of $\sum_{i=0}^{\infty}\chi(i)i=L(-1,\chi)$: just as before, we can formally commute the $a_k$'s in the modes of 'classical energy momentum tensor' to put them into normal order to get the operators $L_m^{\chi}$. As a result, the 'vacuum Casimir energy' will come out naively as the divergent series $\frac{1}{2N}\sum_{i=0}^{\infty}\chi(i)i$. On the other hand, as was mentioned before, physical reasoning restricts the value of this vacuum energy to be the amount by which a shift of the zero mode of the energy momentum tensor has the effect of canceling the linear term in the commutation relations. From the above lemma, we see that the shift (and only this one) $L_0^{\chi}\rightarrow L_0^{\chi}+\frac{1}{2N}L(-1,\chi)$ does this job. We denote $L_0^{\chi}+\frac{1}{2N}L(-1,\chi)$ by $L(')_0^{\chi}$. Then what we have said is
\begin{equation} \label{2.8}
[L_m^{\chi_1},L_{-m}^{\chi_2}]=2mL(')_{0}^{\chi_1\chi_2}+\frac{m^3}{12}\sum_{k=1}^{N}(\chi_1\chi_2)(k)
\end{equation}
Upon canceling the common factor $\frac{1}{2N}$, we see that $\sum_{i=0}^{\infty}\chi(i)i=L(-1,\chi)$ comes out by comparing the values of this same vacuum energy.

As we promised in the end of section 1, let's also see how one gets a similar physical interpretation of axiom $(1)^{'}$ for the series $\sum_{i=0}^{\infty}\chi(i)i$, where $\chi$ is a function satisfying all the requirements stated at the beginning of this section: (Note that this essentially includes the case when $\chi$ is any Dirichlet character, since then $\sum_{i=0}^{\infty}\chi(i)i$ is possibly nonzero only when $\chi$ satisfies all the requirements stated at the beginning of this section.)\\
Having the twisted Virasoro operators set up, the interpretation is exactly the same as for the series $0+1+2+3+...$. Namely, the self-embedding $\tau_{l}$ of the oscillator algebra: $a_k\rightarrow a_{lk}$, $1\rightarrow l$ again tells us that the operators $\frac{1}{l}\tau_l(L^{\chi}_{m})$ satisfy the same Virasoro commutation relations as for the operators $L^{\chi}_{m}$. In other words, the scaling symmetry of the $X^{\mu}$ conformal field theory we talked about generalizes to our twisted Virasoro operators, and this symmetry is again the physical reason for us to impose axiom $(1)^{'}$ to the divergent series $\sum_{i=0}^{\infty}\chi(i)i$, just as before. 

So we have achieved the desired goal as discussed in the end of section 1 by contemplating possible mathematics and physics implications of theorem 1.1. Next we'll talk a bit more about divergent series, namely, how to get $0+1+2+3+...=-\frac{1}{12}$ naturally from our axioms, as we mentioned in section 1. Again, let $\chi$ denote a function satisfying all the requirements stated at the beginning of this section. 

Axiom (2) gives us a linear functional $g$ on a subspace $U$ of the space of series as: we define $U$ to be the space consisting of series for which a finite consecutive sequence of applications of axiom (2) gives a convergent sequence and hence the value of the original series, and the value of $g$ on any series in $U$ is just this value. Obviously $U$ contains all convergent series, for which $g$ gives the usual limit. Furthermore, theorem 1.1 and previous discussions show that when $\chi$ is nontrivial, $\sum_{i=0}^{\infty}\chi(i)i$ is also in $U$, and the value of such series given by $g$ satisfies axiom $(1)^{'}$: in other words, any inflation of $\sum_{i=0}^{\infty}\chi(i)i$ according to the description of axiom $(1)^{'}$ is also in $U$, and $g$ maps it to the same complex number as for $\sum_{i=0}^{\infty}\chi(i)i$ itself. This fact matches nicely with the above physical interpretation of axiom $(1)^{'}$ for such series. Now, it's easy to see that the series $0+1+2+3+...$ is not in $U$. We claim that there exists an extension of the linear functional $g$ to a larger subspace $U'$ containing the series $0+1+2+3+...$, such that the value of $g$ on $0+1+2+3+...$ satisfies axiom $(1)^{'}$ in the above sense. Furthermore, for any such extension, we necessarily have $0+1+2+3+...=-\frac{1}{12}$. The reason is simple: as we have seen in \eqref{**}, for any such extension we must have $0+1+2+3+...=-\frac{1}{12}$. On the other hand, since we know $0+1+2+3+...$ and all its inflations in the sense of axiom $(1)^{'}$ is not in $U$, and furthermore it's obvious that all its inflations are linearly independent, then we may just define a $U'$ by hand to be the smallest subspace of the space of series containing $U$ and all these inflations, and define $g$ on all these inflations to be $-\frac{1}{12}$.

Let's state the above concisely: axiom (2) gives rise to $g$ and $U$, and there exist extensions of $g$ such that it can give value to $0+1+2+3+...$ satisfying axiom $(1)^{'}$, furthermore any such extension necessarily gives the value to be $-\frac{1}{12}$. Note that we have a physical reason for imposing axiom $(1)^{'}$ to $0+1+2+3+...$, so in this sense one gets $0+1+2+3+...=-\frac{1}{12}$ naturally from axiom (2) alone! In other words, if we agree that axiom (2) should hold for divergent series in one's regularization scheme, then we necessarily have $0+1+2+3+...=-\frac{1}{12}$ in the $X^{\mu}$ conformal field theory. Moreoever, the reader can now see that it's a very straightforward generalization to get for general $\chi$, $\sum_{i=0}^{\infty}\chi(i)i=L(-1,\chi)$ and $\sum_{i=0}^{\infty}\chi(i)=L(0,\chi)$ naturally from axioms (2) and $(1)^{'}$, just as the above.                    

\section{Fractional Powers}
Let's denote the Fock space representation of the infinite dimensional Lie algebra generated by the operators $L_0^{\chi}$ (together with the central element) as $Vir^{G}$. By abuse of notation, the same notation $Vir^{G}$ sometimes also mean the Lie algebra itself when there shouldn't be any confusion. We assume $G$ satisfies all conditions mentioned at the beginning of section 2. In this section, we will first analyze the structure of $Vir^{G}$, and it turns out it is as simple as one may possibly expect: it's just a direct sum of several copies of Virasoro algebras sharing the same central element. However, we will show that this Lie algebra interestingly relates some peculiar infinite products with linear combinations of special L values. This Lie algebra relates also to minimal model representations of the Virasoro algebra with negative central charges.
\begin{theorem}
$Vir^{G}$ is isomorphic to a direct sum of $|G|$ copies of Virasoro algebras sharing the same central element.
\end{theorem}
\begin{proof}
We denote $|G|$=$k$. First we do the case when $G$ is cyclic, and then come to the general case.

For the case when $G$ is cyclic, for notational simplicity, let's denote by $1$ as a generator of $G$. Let $\omega$ be a primitive $k$th root of unity. We define operators 
\begin{align}
T_n^{i}=\frac{1}{k}\sum_{s=1}^{k}\omega^{is}L_n^{s}
\end{align}
if $n\neq 0$. Together with
\begin{equation}\label{3.2}
T_0^{i}=\frac{1}{k}\sum_{s=1}^{k}\omega^{is}L(')_0^{s}
\end{equation}
Denote
\begin{equation}
b=\sum_{s=1}^{N}\texttt{id}_G(s)
\end{equation}
We will see that for each $i$, the operators $T_n^{i}$ satisfy
\begin{equation} \label{3.4}
[T_m^{i},T_n^{i}]=(m-n)T_{m+n}^{i}+\delta_{m,-n}\frac{m^3}{12k}b
\end{equation}
and
\begin{equation} \label{3.5}
[T_m^{i},T_n^{j}]=0
\end{equation}
for any $m,n$, if $i\neq j$. So for the Fock space representation of each copy of the Virasoro algebra, the central charge equals $\frac{b}{k}$.

This is because if we look at the coefficient of any $L_{m+n}^{h}$ in $[T_m^{i},T_n^{j}]$, it equals $(m-n)$ times $\frac{1}{k^2}\sum_{s+t\equiv h(\mod k)}\omega^{si+tj}$, which equals $\frac{1}{k}\omega^{h}$ if $i=j$, and $0$ if $i\neq 0$, by elementary property of primitive roots of unity. The coefficient for the central element comes out in the same way if we also keep in mind \eqref{2.8}.

So the theorem is proved when $G$ is cyclic. Now for the general case, since $G$ is finite abelian, we have $G\cong\prod_{i=1}^{d}\mathbb{Z}/m_i\mathbb{Z}$. Let $\chi_i$ be a generator of $\mathbb{Z}/m_i\mathbb{Z}$, and $\omega_i$ be a primitive $m_i$'th root of unity. Then it's easy to see that the direct generalization of the operators $T_n^{i}$:
$$
\frac{1}{k}\sum L_n^{\chi_1^{s_1}\chi_2^{s_2}...\chi_d^{s_d}}\omega_1^{i_1s_1}\omega_2^{i_2s_2}...\omega_b^{i_bs_b}
$$
together with the obvious shift concerning the zero'th mode, give the desired decomposition of $Vir^G$ as a direct sum of $k$ copies of Virasoro algebras sharing the same central element.
\end{proof}
\begin{remark}
It is easy to see that the choice of the operators $T_n^i$ is unique if one wants \eqref{3.4} and \eqref{3.5} both to be satisfied.
\end{remark}
Now we digress to discuss modular transformation properties of some infinite products. Perhaps one of the most well known examples of this type is the Dedekind $\eta$ function
\begin{equation}\label{3.6}
\eta(\tau)=x^{\frac{1}{24}}\prod_{n=1}^{\infty}(1-x^n)
\end{equation}
where
\begin{equation*}
x=e^{2\pi i\tau}
\end{equation*}
$\eta(\tau)$ has famous modular transformation properties which is important in many areas of mathematics, and probably one should be curious about the appearance of the special fractional power $\frac{1}{24}$: why this and only this special power makes the function have the desired modular transformation properties? There are explanations of this, for example, this power can be calculated using the theory of theta functions. However, there is a physical interpretation of this in terms of conformal field theory which is conceptually straightforward:

If we calculate the 'one loop partition function' of the free scalar conformal field theory, the Dedekind $\eta$ function shows up as the main building block, mainly because it is the character (or, some call 'q-trace') of the oscillator representation of the Virasoro algebra. the power $\frac{1}{24}$ shows up exactly because this is the amount of vacuum Casimir energy. For details, one can see \cite{Polchinski}, chapter 7. The one loop partition function is automatically invariant under the modular action of $\operatorname{SL}(2,\mathbb{Z})$ because it should automatically inherit whatever symmetry of the lattice of the relevant elliptic curve. So the Dedekind $\eta$ function should have the desired modular transformation properties derived from the modular invariance property of the one loop partition function. So in this way, one gets a more or less straightforward physical understanding of why the Dedekind $\eta$ function has the desired modular transformation properties. Indeed, only the fractional power $\frac{1}{24}$ can make this miracle happen because $\frac{1}{24}=-\frac{1}{2}\zeta(-1)$ is minus the amount of the vacuum Casimir energy, which is fixed to exactly this value as we said before at the end of section 1. In other words, the fractional power $\frac{1}{24}$ in the Dedekind $\eta$ function can be understood as coming from the vacuum Casimir energy of some conformal field theory, which in turn is given by a special zeta value. For readers interested in mathematical formulation and proof of the above physical intuition, one can see the fundamental paper by Zhu: \cite{Zhu}. Our purpose here is, on the other hand, to explore some implications of this physical intuition to our setting, namely, $Vir^G$. We already see in section 2 that more general divergent series like $\sum_{k=1}^{\infty}\chi(k)k$ can be associated with vacuum Casimir energies of $Vir^G$, which are given by special L values. So a natural question is if some more general fractional powers which appear in some other infinite products having some peculiar modular transformation properties can physically be explained by some sort of vacuum energy, and mathematically given by (linear combinations of ) special L values? The purpose of the rest of section 3 is to answer this question affirmatively in an exact sense. For us, we think this answer does provide a valuable understanding of these fractional powers. (Indeed, a similar physical interpretation is already available if one considers some negative central charge minimal model representations of the Virasoro algebra. However, we don't have explicit constructions of these representations. Furthermore the new relation with special L values is much easier to see in our settings. Lastly our discussions include more general cases. )

In Rogers-Ramanujan identities one considers the curious infinite products
\begin{equation*}
\prod_{n=0}^{\infty}(1-x^{5n+1})(1-x^{5n+4})
\end{equation*}
and
\begin{equation*}
\prod_{n=0}^{\infty}(1-x^{5n+2})(1-x^{5n+3})
\end{equation*}
Furthermore, these two infinite products give the essential part of the characters of minimal model representations of the Virasoro algebra with central charge $c=c_{2,5}=-\frac{22}{5}$. For further details, see equations (3.1),(3.2) of  \cite{Hida}. In particular, Zhu's fundamental theorem applied to this case gives us corollary 1 of theorem 5.2 in \cite{Hida}. For convenience we copy it here:

The complex vector space spanned by the (modified) characters
\begin{equation*}
\overline{ch}_{-\frac{22}{5},0}(q)=q^{\frac{11}{60}}\prod_{n\geq 0}\frac{1}{(1-q^{5n+2})(1-q^{5n+3})}
\end{equation*}
and
\begin{equation*}
\overline{ch}_{-\frac{22}{5},-\frac{1}{5}}(q)=q^{-\frac{1}{60}}\prod_{n\geq 0}\frac{1}{(1-q^{5n+1})(1-q^{5n+4})}
\end{equation*}
is modular invariant.

More generally, for every integer $k\geq 2$, there are exactly $k$ inequivalent minimal model representations of the Virasoro algebra, with central charge 
\begin{equation}
c_{2,2k+1}=1-\frac{6(2k-1)^2}{4k+2}
\end{equation}
, and highest weight 
\begin{equation}
h_{2,2k+1}^{1,i}=\frac{(2(k-i)+1)^2-(2k-1)^2}{8(2k+1)}
\end{equation}
, where $i=1,2,...,k$.\\
The (modified) characters are
\begin{equation}\label{3.9}
\overline{ch}_{c_{2,2k+1},h_{2,2k+1}^{1,i}}(q)=q^{h_{2,2k+1}^{1,i}-\frac{c_{2,2k+1}}{24}}\prod_{n\neq \pm i,0 \mod{2k+1}}\frac{1}{(1-q^n)}
\end{equation}
Furthermore the vector space spanned by these characters is modular invariant.

Next, we will show how to express these infinite products in terms of theta functions and prove the above modular invariance property by the theory of theta functions( this is more or less a routine exercise of the theory of theta functions, but we include it here for later convenience).Our notations for theta functions are according to \cite{theta}.\\
First we have the definition of theta function with characteristic $[\frac{\epsilon}{\epsilon'}]\in\mathbb{R}^2$
\begin{equation}\label{theta}
\theta[\frac{\epsilon}{\epsilon'}](z,\tau)=\sum_{n\in\mathbb{Z}}\exp2\pi i\left\{\frac{1}{2}(n+\frac{\epsilon}{2})^2\tau+(n+\frac{\epsilon}{2})(z+\frac{\epsilon'}{2})\right\}
\end{equation}
which converges uniformly and absolutely on compact subsets of $\mathbb{C}\times\textbf{H}^2$.\\
Next we have Euler's identity
\begin{equation}\label{3.11}
\prod_{n=1}^{\infty}(1-x^n)=\sum_{n=-\infty}^{\infty}(-1)^nx^{\frac{n(3n+1)}{2}}
\end{equation}
As is well known, from the above two equations we can easily express the Dedekind $\eta$ function in terms of theta function
\begin{equation}\label{3.12}
\eta(\tau)=e^{-\frac{\pi i}{6}}\theta[\frac{\frac{1}{3}}{1}](0,3\tau)
\end{equation}
Moreover, we have the Jacobi triple product identity
\begin{equation}\label{3.13}
\prod_{n=1}^{\infty}(1-x^{2n})(1+x^{2n-1}z)(1+\frac{x^{2n-1}}{z})=\sum_{n=-\infty}^{\infty}x^{n^2}z^n
\end{equation}
for all $z$ and $x$ in $\mathbb{C}$ with $z\neq 0$ and $|x|<1$.

In the above equation, substitute $x$ by $x^{\frac{2k+1}{2}}$, and $z$ by $-x^\frac{2k+1-2j}{2}$, we get
\begin{equation}\label{3.14}
\prod_{n=1}^{\infty}(1-x^{(2k+1)n})(1-x^{(2k+1)n-j})(1-x^{(2k+1)n-(n-j)})=\sum_{n=-\infty}^{\infty}x^{\frac{2k+1}{2}n^2+\frac{2k+1-2j}{2}n}(-1)^n
\end{equation}
In \eqref{theta}, take $\epsilon=\frac{2(k-j)+1}{2k+1}$, and $\epsilon'=1$, we get
\begin{equation}\label{3.15}
\theta[\frac{\frac{2(k-j)+1}{2k+1}}{1}](0,(2k+1)\tau)=\sum_{n\in\mathbb{Z}}x^{\frac{2k+1}{2}n^2+\frac{2(k-j)+1}{2}n}x^{\frac{(2(k-j)+1)^2}{8(2k+1)}}(-1)^ne^{\frac{2(k-j)+1}{2(2k+1)}\pi i}
\end{equation}
We combine \eqref{3.14} and \eqref{3.15} to express the left hand side of equation \eqref{3.14} in terms of theta functions, and then divide the result by \eqref{3.12} and use \eqref{3.6}, we get
\begin{equation}\label{3.16}
\prod_{s\neq \pm j,0 \mod{2k+1},n\geq 1}\frac{1}{(1-x^{(2k+1)n-s})}=x^{\frac{1}{24}-\frac{(2(k-j)+1)^2}{8(2k+1)}}e^{\frac{\pi i}{6}-\frac{\pi i(2(k-j)+1)}{2(2k+1)}}\frac{\theta[\frac{\frac{2(k-j)+1}{2k+1}}{1}](0,(2k+1)\tau)}{\theta[\frac{\frac{1}{3}}{1}](0,3\tau)}
\end{equation}
Note that the power of $x$ that shows up in the above equation exactly equals to $-(h_{2,2k+1}^{1,j}-\frac{c_{2,2k+1}}{24})$. So combining \eqref{3.16} and \eqref{3.9}, we see that the complex vector space spanned by the (modified) characters $\overline{ch}_{c_{2,2k+1},h_{2,2k+1}^{1,j}}$ is the complex vector space spanned by the functions
\begin{equation*}
\frac{\theta[\frac{\frac{2(k-j)+1}{2k+1}}{1}](0,(2k+1)\tau)}{\theta[\frac{\frac{1}{3}}{1}](0,3\tau)}
\end{equation*}
where $j$ runs from $1$ to $k$. A direct application of Lemma 4.2 on page 216 of \cite{theta} gives the theta function theory proof that this vector space is modular invariant. 

Now we come back to our $Vir^G$. For simplicity, Let's first assume that $N=2k+1$ is an odd prime, and let $G$ be the group of Dirichlet characters of conductor $N$ which maps $-1$ to $1$. So $G$ is cyclic of order $k$. Let $\chi$ be a generator of $G$. We first calculate the vacuum Casimir energies associated with the $T_0^i$'s. Let's denote this quantity by $c_i$. Namely, we express $T_0^i$ as a linear combination of the unshifted operators $L_0^i$'s and a constant. The amount of the vacuum Casimir energy associated with $T_0^i$ is just this constant term. Recall \eqref{3.2} and \eqref{L1}, we have
\begin{equation}
c_i=\frac{1}{2k}\sum_{s=1}^{k}\omega^{is}L(-1,\chi^s)
\end{equation}
since $G$ satisfies the assumptions at the beginning of section 2, and for the trivial character $\chi^k$, we have
\begin{equation}
\sum_{k=1}^{N}\chi(k)=N-1
\end{equation} 
then an easy calculation shows
\begin{equation}
c_i=\frac{1}{2k}(-\frac{N}{12}(N-1))+\frac{1}{2k}(k\frac{-j^2-(N-j)^2}{2N}+k\frac{j+N-j}{2})
\end{equation}
Where $j,N-j$ is the unique pair such that 
\begin{equation}\label{3.20}
\chi(j)=\omega^{k-i}
\end{equation}
In other words,
\begin{equation}\label{3.21}
c_i=\frac{2k+1}{12}-\frac{j(N-j)}{2N}
\end{equation}
Since $Vir^G$ contains $k$ copies of the Virasoro algebra, we have $k$ concepts of vacuums defined by\\
$T_0^i$ vacuum:
\begin{equation}
T_0^i=0
\end{equation}
$i=1,2,...,k$.
These vacuums are 'orthogonal' in the sense of \eqref{3.5}, and each vacuum contains an infinite number of degenerate states in the Fock space. Suppose now we pick any $i$ and consider the $T_0^i$ vacuum. This is a subspace of the Fock space on which the operators $L_m^k-T_m^i$ act, since from \eqref{2.8} and \eqref{3.4} we can easily verify
\begin{equation}\label{3.23}
[L_m^k-T_m^i, T_0^i]=0
\end{equation}
for any $m\in\mathbb{Z}$.

Now we calculate the 'q-trace' of $L_0^k-T_0^i$ on the $T_0^i$ vacuum. More precisely, we should calculate this quantity for the operator corresponding to $L_0^k-T_0^i$ before we do any shift by vacuum energies, just as what one does for the usual free scalar conformal field theory. In other words, we need to include the effect of the vacuum energy associated to the operator $L_0^k-T_0^i$. Here we explain some intuitive reasons for doing this:

It is easy to see that $T_0^i$ contains only oscillator modes that are congruent to $j$ or $N-j$ $\mod N$, where $j$ satisfies \eqref{3.20}, and the $T_0^i$ vacuum consists exactly of states in the Fock space without these oscillator modes. From the physical point of view, the vacuum energy is the sum of zero point energies of all relevant oscillator modes. So whenever we want to calculate some sort of vacuum energy in the $T_0^i$ vacuum, we should erase the effect of these oscillator modes.

The vacuum energy associated to $L_0^k-T_0^i$ is
\begin{equation}
d_i=\frac{1}{2}L(-1,\chi^k)-c_i
\end{equation}
From \eqref{L1} one can easily obtain
\begin{equation}
L(-1,\chi^k)=\frac{N-1}{12}
\end{equation}
So from the above equation and \eqref{3.21}, we obtain
\begin{equation}\label{3.26}
d_i=\frac{(2(k-j)+1)^2}{8(2k+1)}-\frac{1}{24}
\end{equation}
which is easily checked to be equal to $h_{2,2k+1}^{1,j}-\frac{c_{2,2k+1}}{24}$.

Furthermore, recall that $T_0^i$ contains only oscillator modes that are congruent to $j$ or $N-j$ $\mod N$, and $L_0^k$ misses only oscillator modes that are divisible by $N$, an easy calculation shows that
\begin{equation}
\text{the (shifted) 'q-trace' of $L_0^k-T_0^i$ on the $T_0^i$ vacuum}=\overline{ch}_{c_{2,2k+1},h_{2,2k+1}^{1,j}}
\end{equation}
So as we said before, the powers of $x$ given by \eqref{3.26} or $h_{2,2k+1}^{1,i}-\frac{c_{2,2k+1}}{24}$ which are crucial for the infinite products in \eqref{3.9} to have special modular transformation properties, are explained by this vacuum Casimir energy, and consequently are expressed explicitly as linear combinations of special L values as
\begin{equation}\label{3.28}
\frac{(2(k-j)+1)^2}{8(2k+1)}-\frac{1}{24}=h_{2,2k+1}^{1,i}-\frac{c_{2,2k+1}}{24}=\frac{1}{2}L(-1,\chi^k)-\frac{1}{2k}\sum_{s=1}^{k}\omega^{is}L(-1,\chi^s)
\end{equation}
Remember that the above calculation is based on the assumption that $N$ is an odd prime. In general, for $N$ not necessarily prime, instead of Dirichlet characters, we may use the group $G$ of functions $\textbf{N}\rightarrow\mathbb{C}$ of period $N$ defined by
\begin{equation}
f_s(u)=\theta^{su}
\end{equation}
for $u=1,2,...,k$, $s=1,2,...,k$, and
\begin{equation}
f_s(N-u)=f_s(u)
\end{equation}
and
\begin{equation}
f_s(N)=0
\end{equation}
where $\theta$ is a primitive $k$th root of unity.

It's straightforward to see that $G$ satisfies all the assumptions made at the beginning of section 2, and all the above calculations work out without change. (At this stage it is crucial that formulas such as \eqref{L1} work for this kind of more general functions just the same as for Dirichlet characters.)

Note that for the usual oscillator representation of the Virasoro algebra, the corresponding 'q-trace' gives the Dedekind $\eta$ function which gives rise to modular forms for the full modular group $\operatorname{SL}(2,\mathbb{Z})$, which is what to be expected from a physical point of view since the one loop partition function should inherit all the symmetry of the lattice defining an elliptic curve (over $\mathbb{C})$. If our construction of $Vir^G$ is to be an analogue of a usual conformal field theory in a deeper sense, then the 'q-traces' should have similar modular properties, which is what we have shown to be true by using the theory of theta functions. However, it is obvious that a single 'q-trace' like the 'q-trace' of $L_0^k-T_0^i$ on the $T_0^i$ vacuum, which is equal to $\overline{ch}_{c_{2,2k+1},h_{2,2k+1}^{1,j}}$, possesses nice modular transformation properties only for a certain subgroup of finite index of the modular group. Only a collection of these, namely, the complex vector space spanned by all these similar 'q-traces', are invariant under the full modular group. To be more precise, it's not hard to see (yet not very straightforward) that up to a constant, the 'q-trace' of $L_0^k-T_0^i$ is invariant under the action of any element of the principal congruence subgroup $\Gamma(2k+1)$. (To see this, we may apply equation (2.16) on page 81 of \cite{theta}. There are two crucial facts needed to verify this: the equivalent class of the characteristic $[\frac{\frac{2l+1}{2k+1}}{1}]$ is invariant under the action of $\Gamma(2k+1)$, and $\Gamma(2k+1)$ is contained in the Hecke subgroup $\Gamma_0(2k+1)$.) On the other hand, it is well known that the modular curve $X(2k+1)$ given by $\Gamma(2k+1)$ is the moduli space for elliptic curves with a given basis for the $2k+1$ torsion. So it seems reasonable for one to suspect that at one loop level, twisting by our cyclic group $G$ reflects on the geometric side as the additional information of a given basis for the $2k+1$ torsion. This is interesting but we don't yet know how to make this precise.

At the end of this section, we elaborate a bit on our calculations relating $Vir^G$ with those fractional powers, and we ask a question. From \eqref{3.11} through \eqref{3.15}, it is straightforward to express $(1-x^{(2k+1)n-j})(1-x^{(2k+1)n-(n-j)})$ as a constant times a power of $x$ times a quotient of theta functions. The power that shows up is equal to
\begin{equation}\label{3.32}
-\frac{(N-2j)^2}{8N}+\frac{N}{24}
\end{equation}
which equals $-c_i$ by \eqref{3.21}.

On the other hand, similar to \eqref{3.23}, we have also
\begin{equation}\label{3.33}
[L_0^k-T_0^i,T_n^i]=0
\end{equation}
for all $n\in\mathbb{Z}$. So we can calculate the character of the representation of the Virasoro algebra given by the operators $T_n^i$ on the subspace $L_0^k-T_0^i =0$. Taking into account the vacuum energy associated with $T_0^i$, it's easy to see that this character (or 'q-trace') is given by
\begin{equation}\label{3.34}
x^{\frac{(N-2j)^2}{8N}-\frac{N}{24}}(1-x^{(2k+1)n-j})(1-x^{(2k+1)n-(n-j)})
\end{equation}
Comparing \eqref{3.32} and \eqref{3.34}, we see the relation between $Vir^G$ an these fractional powers also works in this other direction where we consider also the 'q-trace' of $T_0^i$'s on the subspace $L_0^k-T_0^i =0$, which now is indeed a character of a representation of the Virasoro algebra.

Although the theory of theta functions can be used to prove some of the 'q-traces' of our $Vir^G$ having peculiar modular transformation properties, we don't think this kind of proof is satisfactory. Here we want to pose the question of whether one can extend Zhu's fundamental result in this direction \cite{Zhu} to give a uniform 'conceptual' proof of the modular properties of these 'q-traces'. Despite the simplicity of $Vir^G$, one already sees that it connects several things. We hope that investigation of this question will have implications for what $Vir^G$ really means, and is it indeed an interesting tool connecting conformal field theory with number theory.

\section{Discussion}
In this section, we discuss some possibilities and problems trying to generalize our story of $Vir^G$. Also we try to indicate some connections between class field theory and our $Vir^G$.

First of all, note that there is good reason to impose on $\chi$ the condition given by \eqref{2.1}: since for the case when $\chi$ is a Dirichlet character, we have either \eqref{2.1}, or $\chi(-1)=-1$. For the latter case, it's easy to see that if we still make the same definition for $L_n^{\chi}$ as in \eqref{2.3}, then $L_n^{\chi}=0$ for all $n\neq 0$. It's well known that in this case, $L(-1,\chi)=0$ (this can easily be seen from the functional equation of $L(s,\chi)$, and that $L(s,\chi)$ is always analytic at $s=2$). So the theory in this case is trivial.

Next we make an observation connecting field theory with our $Vir^G$:

Let $K$ be a totally real finite abelian extension of $\mathbb{Q}$. Then we have
\begin{theorem}
$K$ gives rise to a unique $Vir^G$ in a canonical way.
\end{theorem}
\begin{proof}
It's well known that there is a one to one correspondence between subfields of cyclotomic fields and groups of Dirichlet characters, and regarding $K$ as subfields of different cyclotomic fields give rise to the same group of Dirichlet characters. Furthermore, $K$ being totally real implies that the corresponding group of Dirichlet characters is even, i.e., it satisfies \eqref{2.1}, so the theorem follows.
\end{proof}
Next let us point out some obstacles stopping us from straightforwardly generalizing our construction of $Vir^G$ to some other settings.

First of all, the Sugawara construction and the GKO construction are more or less direct generalizations of the oscillator representation of the Virasoro algebra, so it may be very natural to try to generalize the construction of $Vir^G$ to those representations. However, as we tried, direct generalization does not work. What stops us is exactly the nonzero Dual Coxeter number $g$ of nonabelian (finite dimensional) Lie algebras. However, there is good physical reason explaining why this doesn't work: the GKO construction corresponds to the conformal field theory of some currents associated with a (nonabelian) Lie group (or Lie algebra) symmetry, which have conformal dimensions $1$. So we no longer have the scaling symmetry as that for the $X^{\mu}$ fields. If the direct generalization of $Vir^G$ were to work as before, then we ought to have values of certain divergent series given by certain special L values as before. So conversely our axiom $(1)^{'}$ would better be put into work again. But axiom $(1)^{'}$ is a mathematical reflection of the scaling symmetry of the physics, which is lost.

Secondly, let's consider the seemingly more delicate possibility of generalizing the construction of $Vir^G$ to some number fields other than just $\mathbb{Q}$.

There is an obvious reason in trying to do this: we have successfully constructed $Vir^G$ for Dirichlet characters and consequently obtained some relations between special L values of Dirichlet characters and vacuum energies, and those fractional powers. Since Dirichlet characters are equivalent to Hecke characters of finite order for the field $\mathbb{Q}$, so it seems natural to think if it is possible to generalize the same construction to some more general number fields, replacing Dirichlet characters by Hecke characters on the ideles, and Dirichlet L functions by Hecke L functions, hoping to have similar stories going on. However, we are not succeeded in doing this at least in a direct way. From a very practical point of view, we think one of the major problems is that for number fields other than $\mathbb{Q}$, the rings of integers are too big for us to do any meaningful normal ordering. From a more 'conceptual' point of view, our story being worked for $\mathbb{Q}$ crucially relies on the fact that we know what a conformal field theory is for the genus zero case: it's described by vertex operator algebra theory. In particular, the conformal vector gives us a copy of the Virasoro algebra, which is where we put our hands on. So it seems reasonable to suspect that our difficulty here is probably tied to the difficulty of defining conformal field theory on higher genus Riemann surfaces. Last but not least, as we tried, not surprisingly, technical difficulties include the problem of dealing with global units and nonzero class number of the number field (and possibly archimedean places). We end our discussion for this type of questions here with a sort of 'no-go' lemma concerning the Virasoro algebra.

Again let $K$ be a number field, and $\textsl{O}_K$ be its ring of integers. We fix an arbitrary embedding of $K$ into $\mathbb{C}$, and regard $K$ as a subfield of $\mathbb{C}$. (We will see later that this choice is inessential) Suppose for each $n\in \textsl{O}_K$ we have an operator denoted by $L_n$. In the complex linear span of these operators together with a central element $c$, we suppose we have the following commutation relations:
\begin{equation}\label{4.4}
[L_m,L_n]=(m-n)L_{m+n}+\alpha(m,n)c
\end{equation}
for any $m,n\in \textsl{O}_K$. Where $\alpha(m,n)$ is a function from $\textsl{O}_K\times\textsl{O}_K$ to $\mathbb{C}$.
Then we have
\begin{lemma}  
The above defines a Lie algebra if and only if $\alpha(m,n)$ satisfies
\begin {equation}\label{4.5}
\alpha(m,n)=\delta_{m,-n}\alpha(m)
\end{equation}
\begin{equation}\label{4.6}
\alpha(-m)=-\alpha(m)
\end{equation}
\begin{equation}\label{4.7}
\alpha(m)=am+bm^3
\end{equation}
for some $a,b\in \mathbb{C}$.
\end{lemma}
\begin{remark}
Obviously \eqref {4.6} is a special case of \eqref{4.7}, but we include it here for convenience. Another issue to point out is that to define a complex Lie algebra like the above (with or without central extension) requires a choice of the embedding of $\textsl{O}_K$ into $\mathbb{C}$, however, both the Lie bracket and the relevant second Lie algebra cohomology here are natural, with respect to different choices of embeddings into $\mathbb{C}$. Also note that one may consider the question of the above lemma in a more general situation if one wants to.
\end{remark} 
\begin{proof}
The 'if' part is obvious. So we only prove for the 'only if' part. For the first part of the proof, we proceed as on page 8 of \cite{Kac}. In particular, we can get the first two properties of $\alpha(m,n)$,i.e., \eqref{4.5} and \eqref{4.6} in the same way, and also
\begin{equation}\label{4.8}
(m-n)\alpha(m+n)-(2n+m)\alpha(m)+(n+2m)\alpha(n)=0
\end{equation}
Now we pick any integral basis $b_1,b_2,...,b_s$ of $\textsl{O}_K$. We first show that $\alpha(b_1)$ and $\alpha(b_2)$ alone determine the value of $\alpha$ on the $\mathbb{Z}$ linear span of $b_1$ and $b_2$: from \eqref{4.8} we see that from $\alpha(b_1)$ and $\alpha(b_2)$ we can get the values of $\alpha$ also on $b_1+b_2$ and $b_1-b_2$. So by \eqref{4.8} again, these determine the values of $\alpha$ on $(b_1+b_2)+(b_1-b_2)=2b_1$ and $(b_1+b_2)-(b_1-b_2)=2b_2$. Now we put $m$ to be $mb_1$, and $n$ to be $b_1$ in \eqref{4.8}, we have
\begin{equation}\label{4.9}
(m-1)\alpha((m+1)b_1)=(m+2)\alpha(mb_1)-(2m+1)\alpha(b_1)
\end{equation} 
upon canceling the common factor $b_1$.

From the above equation and \eqref{4.6} we see that $\alpha(b_1)$ and $\alpha(2b_1)$ determine $\alpha(mb_1)$ for any $m\in\mathbb{Z}$, and the same for $b_2$. Again, for any $mb_1+nb_2$, the value of $\alpha$ on which is determined by $\alpha(mb_1)$ and $\alpha(nb_2)$, which in turn are determined by $\alpha(b_1)$ and $\alpha(b_2)$, as we have shown. 

On the other hand, we observe that $\alpha(m)=m$ and $\alpha(m)=m^3$ are two linearly independent solutions to \eqref{4.8}. So on $\mathbb{Z}$ linear span of $b_1$ and $b_2$, \eqref{4.7} has to be hold for some complex numbers $a$ and $b$.

Next we consider $b_1$ and $b_3$, the same argument deduces that on $\mathbb{Z}$ linear span of $b_1$ and $b_3$, \eqref{4.7} has to be hold for some complex numbers $a'$ and $b'$. So in particular, on the $\mathbb{Z}$ linear span of $b_1$, \eqref{4.7} has to be hold for $a$ and $b$, and also $a'$ and $b'$. So $a=a'$, $b=b'$. So \eqref{4.7} holds for $b_3$ for the same complex numbers $a$ and $b$. Extending this argument, we see \eqref{4.7} holds for any element in our chosen integral basis. Then by \eqref{4.8} again, we see it holds for any element in $\textsl{O}_K$.
\end{proof}
\begin{remark}
By choosing an appropriate additive semigroup inside $\textsl{O}_K$, one may define highest weight representations and Verma modules of this Lie algebra. However, after axiomatizing some of the essential ingredients of the Oscillator representation, one can show that oscillator-like representations do not exist if the field $K$ is not $\mathbb{Q}$. The details of these will not be treated in this paper.
\end{remark}   
\section{A physical Interpretation of Tate's Thesis}
Now it's time for us to end this paper by formulating a relation between Tate's thesis and conformal field theory as we promised. For introductory material on Tate's thesis, one can see for example \cite{Kudla}.

In \cite{Witten}, Witten formulated several quantum field theories on an (smooth, complete) algebraic curve over an algebraically closed field. Here we will try to formulate a simplest possible conformal field theory on an algebraic number field from a somewhat different point of view. We will use some ideas of \cite{Witten}, of course. We will take some analogues of these ideas and apply them to the case of number fields for guidance. We have no intention to make our discussion here rigorous or complete, however. Our goal here is to tentatively explore this possible connection between number theory and physics. We will see that much of Tate's thesis come out from physical considerations.

Let $K$ be a number field, and $\textsl{O}_K$ its ring of integers. Let $A_K$ be the ring of adeles, and $I_K$ the idele group, $C_K$ the idele class group, and $\textsl{I}_K$ the ideal class group. We denote by $\tau$ the diagonal embedding of $K^{\times}$ into $I_K$. We fix a global additive character $\psi$ of $A_K$, trivial on $K$. For any local embedding $F$ of $K$, let $d^{\times} x$ denote the multiplicatively invariant Haar measure on $F$ normalized so that the (local) units have volume $1$. Also we denote by $dx$ the self-dual additively invariant Haar measure with respect to the local component of $\psi$. By abuse of notation, we also denote by $d^{\times} x$ (and $dx$) the multiplicatively (and additively) invariant Haar measure on $I_K$ given by multiplying the local Haar measures. Now we will attempt to describe what one may call the $\operatorname{GL}(1)$ 'current group' on a number field.

First of all, for a commutative ring, we have at our hand the geometric object given by the prime spectrum of the ring, to be used to take analogue with the geometric case. For any place $v$ of $K$, local operators are in $\text{Hom}(\text{Spec} K_v, \operatorname{GL}(1))=\operatorname{GL}(1,K_v)$. By taking analogue of the discussion on multiplicative Ward identities in \cite{Witten}, if the local operator $f_v$ has negative valuation, then physically it corresponds to a positive energy excitation at $v$. So globally, quantum fields live in $\prod_{v}\operatorname{GL}(1,K_v)$, with the restriction that just like ordinary conformal field theory, for all but finitely many places $v$, $f_v$ lives in $\operatorname{GL}(1,O_v)$. So, in other words, quantum fields are elements of the idele group $I_K$.

Next, any two quantum fields differing by an element of $\tau(K^{\times})$ should be regarded as the same. We have reasons for imposing this requirement: one may consult section V of \cite{Witten}. Multiplying by elements of $\tau(K^{\times})$ is the analogue of conformal symmetry transformation.

So the path integral should be on the idele class group $C_K$. To integrate, we need a measure which should be an analogue of what physicists call the Feynman measure on the space of fields. In ordinary quantum field theory on flat spacetime, this (undefined) concept of Feynman measure should be translational invariant, which can be regarded as a consequence of the symmetries of flat spacetime. In our multiplicative case, the analogue of this is the requirement that the measure should be multiplicatively translational invariant. So this measure has to be the Haar measure on $C_K$ with an ambiguity of a scalar, which makes perfect sense.

Next, in path integral formulation of ordinary quantum mechanics and quantum field theory, expressions like $e^{iHt}$, $e^{i\int Ldt}$, or $e^{\int Ldx}$ show up essentially because of the Schrodinger equation, which itself can be regarded more or less as a consequence of the basic principles of quantum mechanics and the flat spacetime Lorentz symmetry. (There are many discussions on this issue, and we won't discuss it here. Note that the Schrodinger equation itself is not Lorentz invariant.) Here on the ideles, we have the multiplicative translational symmetry for the Haar measure, so what substitutes $e^{\int Ldx}$ in the path integral should be a multiplicative function on $C_K$ (Note that formally, $e^{iHt}$ is a quasicharacter on the additive group of $t$, which is a consequence of the Schrodinger equation.), which is nothing but $\omega\omega_s$ in general (on physics grounds, we assume that this function should be continuous. Furthermore we will provide physical interpretations of $\omega$ and $\omega_s$ in our coming paper on reciprocity laws), where $\omega$ is a Hecke character on $I_K$, and $\omega_s$ is the quasicharacter given by
\begin{equation}
\omega_s(x)=|x|^s
\end{equation}
for any $x\in I_K$. Where $s$ is a complex number.

Note that $\omega$ and $\omega_s$ can be factorized as products of local characters, and this is consistent with integrating the Lagrangian density over spacetime in ordinary quantum field theory (or over the worldsheet in two dimensional conformal field theory).\

Before we go any further, let us stop and make an observation which gives us a hint that our construction possibly can come from a gauge theory. For an ordinary $U(1)$ gauge theory, the path integral should sum over all possible $U(1)$ principal bundles over the base manifold. Here, we have the canonical isomorphism
\begin{equation}\label{4.11}
\mathrm{Pic}(\texttt{Spec} \textsl{O}_K)\cong\textsl{I}_K
\end{equation}
Where the Picard group $\mathrm{Pic}(\texttt{Spec} \textsl{O}_K)$ classifies the isomorphism classes of invertible sheaves on $\texttt{Spec} \textsl{O}_K$. In fact, our path integral on $C_K$ somehow sums over $\textsl{I}_K$:\\
In number theory, there is a canonically defined surjective group homomorphism from the idele class group to the ideal class group: 
\begin{equation*}
\pi: C_K\rightarrow \textsl{I}_K
\end{equation*}
with
\begin{equation*}
Ker\pi=I(S_{\infty})/\tau(R^{\times})
\end{equation*}
where
\begin{equation*}
I(S_{\infty})=\prod_{\text{archimedean places}} K_v^{\times}\times \coprod_{\text{nonarchimedean places}} R_v^{\times}
\end{equation*}
where $R_v^{\times}$ is the group of (local) units in $\textsl{O}_v$, and $R^{\times}$ is the group of global units of $\textsl{O}_K$. So integration over $C_K$ already includes a summation over $\textsl{I}_K$. $\pi$ refines the information in the ideal class group, whose usefulness is illustrated by global class field theory. For us, it's usefulness is revealed by the path integral.

Before we can write down the path integral, we still have to consider the insertion of local operators. In ordinary quantum field theory, we have expressions like
\begin{equation}
\int \phi(x)e^{\int L(\phi(x))d^Dx}D\phi
\end{equation}
However, it is hard to make sense of it unless one makes the inserted operators have good decaying properties, and thinks of the measure as a linear map from some space of functions to $\mathbb{R}$. See \cite{Borcherds} for discussions on this issue.

To integrate over $C_K$, the integrand should be functions on $C_K$. Of course, the insertion of local operators should carry appropriate physical meaning. To think about what is the form our insertion should look like, here we consult the form of Polyakov path integral. See for example, \cite{Polchinski}, equation (3.5.5): for the inclusion of a particle, one inserts in the path integral a local vertex operator given by the state-operator correspondence. Furthermore, to make the vertex operator insertions diff-invariant, one integrates them over the worldsheet.

To mimic this process of insertion of vertex operators, we start from an unkown function $f(x)$ on $I_K$ which is a product of local functions with suitable decaying properties, and carrying appropriate physical meaning. Then we sum over $K^{\times}$ to make it $K^{\times}$ invariant (so we insist that $f(x)$ should make the following sum convergent):
$$
\sum_{\alpha\in K^{\times}} f(\alpha x)
$$
\begin{remark}
From the above, it is not correct to say that $K^{\times}$ should be the analogue of the string worldsheet, since the Polyakov path integral is intended to calculate string S-matrices, whereas our path integral is to be regarded as a path integral in conformal field theory. Rather, it makes some sense to regard $\texttt{Spec} \textsl{O}_K$ as the analogue of the worldsheet. But archimedean places should also matter, so one should really say that the theory is to live on number fields.
\end{remark} 
Finally we can write down our path integral:
\begin{equation}
\int_{I_K/\tau(K^{\times})}\sum_{\alpha\in K^{\times}} f(\alpha x)(\omega\omega_s)(x)d^{\times}x
\end{equation}
Furthermore, note that $(\omega\omega_s)(\alpha x)=(\omega\omega_s)(x)$, for any $\alpha\in K^{\times}$. So the above equals
\begin{equation} 
\int_{I_K}f(x)(\omega\omega_s)(x)d^{\times}x
\end{equation}
which is exactly the global zeta integral $\textsl{z}(s,\omega;f)$ for the test function $f$. So we propose that the allowed functions should be in $S(A_K)$, the space of Schwartz-Bruhat functions on $A_K$.

Note that in the above integral, we have only one parameter $s$ which can be continuously varied. So it's tempting to regard $s$ as coming from the 'coupling constant'. We will see what this means as a coupling constant in our coming paper on reciprocity laws.

If we allow $f$ to vary, the global zeta integral $\textsl{z}(s,\omega)$ becomes a distribution, which is well known to be convergent for $\Re(s)>1$, and has a meromorphic analytic continuation to the whole $s$ plane and satisfies the functional equation
\begin{equation}\label{4.17}
\widehat{\textsl{z}(1-s,\omega^{-1})}=\textsl{z}(s,\omega)
\end{equation}
Where the global Fourier transform $\widehat{}$ is defined after we fix the additive character $\psi$ of $A_K$. We also have the functional equation for the complete global L function
\begin{equation} 
\Lambda(s,\omega)=\epsilon(s,\omega)\Lambda(1-s,\omega^{-1})
\end{equation}
which is independent of the choice of $\psi$.

\eqref{4.17} tells us that we can use analytic continuation to define our quantum theory for any value of the coupling constant $s$. We will make use of this fact in our paper discussing reciprocity laws. 
\begin{remark}
If we switch from number fields to global function fields (namely, function fields over a finite field), since Tate's thesis works for both cases, all the above discussion is essentially valid, except that we don't need to worry any more about archimedean places, and also we don't need to take analogues between number fields and function fields. (For the global function field case, we also have a canonical group homomorphism from the idele class group to the divisor class group, which should replace our discussion above around \eqref{4.11}. ) It is interesting to note that \cite{Witten} discusses quantum field theories on curves over an algebraically closed field, where Witten uses algebraic constructions relying on the algebraically closedness of the ground field, and he also remarks: " While one would wish to have an analogue of Lagrangians and quantization of Lagrangians in this more general setting, such notions appear rather distant at present." On the other hand, if we apply our discussion to global function field case, we are actually discussing conformal field theory on curves over a finite field. What we were trying to do, was just to write down a path integral which mimics a path integral in ordinary quantum field theory. But our discussion is not valid for curves over algebraically closed fields.
\end{remark}
\section*{Acknowledgments}
I would like to thank my advisor Richard Borcherds for his advising throughout the years, and in particular on this paper. I also thank Chul-hee Lee, Shenghao Sun and John Baez for helpful discussions and correspondence. Lastly but importantly, I thank Royce Cheng-Yue for carefully catching grammatical errors. I am responsible for all possible remaining errors or typos.

\bibliographystyle{amsplain}

\end{document}